\documentclass[11pt]{article}

\usepackage{times}
\usepackage{amssymb,epsfig,verbatim}
\topmargin by -\headsep \textheight 8.9in \oddsidemargin 0pt
\evensidemargin \oddsidemargin \textwidth 6.5in
\def\Z2{{\mathbb Z}_2}

\def\cl{\mbox{\it cl}\kern.2ex}
\def\C{{\mathcal C}\kern.1ex}
\def\,{\kern.15ex}

\def\beginpr{\begin{proof}}
\def\endpr{$\square$\end{proof}}

\newtheorem{thm}{Theorem}[section]
\newtheorem{lmma}[thm]{Lemma}
\newtheorem{corr}[thm]{Corollary}
\newtheorem{define}[thm]{Definition}
\newenvironment{theorem}
  {\begin{thm}\hskip-6pt{\bf .}\enspace \sl}{\end{thm}}

\newenvironment{proof}
  {\vskip5pt\hspace*{15pt}{\it Proof}.\hskip10pt}{\qed\vskip5pt}
\def\qed{\nobreak\hskip1pt~$\;\;\scriptstyle\Box$}

\begin{document}

\title{Constrained multilinear detection \\ for faster functional motif discovery}


\author{Ioannis Koutis \thanks{This work has been partially supported by NSF CAREER Award 1149048.} \\
Computer Science Department\\[-3pt]
University of Puerto Rico, Rio Piedras\\[-3pt]
ioannis.koutis@upr.edu}

\maketitle

%

\begin{abstract}
The \textsc{Graph Motif} problem asks whether a given multiset of colors appears on a connected subgraph of a vertex-colored graph. The fastest known parameterized algorithm for this problem is based on a reduction to the $k$-Multilinear Detection (\textsc{$k$-MlD}) problem: the detection of multilinear terms of total degree $k$ in polynomials presented as circuits.
We revisit \textsc{$k$-MlD}  and define \textsc{$k$-CMlD}, a constrained version  of it which reflects \textsc{Graph Motif} more faithfully. We then give a fast algorithm for \textsc{$k$-CMlD}. As a result we obtain faster parameterized algorithms for \textsc{Graph Motif} and variants of it.
\end{abstract} 
\section{Introduction}

The \textsc{Graph Motif} problem was introduced in \cite{LacroixFS06} in the context of metabolic network analysis. It asks whether a given vertex-colored graph contains a connected subgraph whose colors agree with a given multiset of colors, the `motif'. In this context motifs are often described as `functional', to distinguish from `topological' motifs, which are fixed subgraphs such as a $k$-path. Further applications of functional motif discovery have been discussed in \cite{BetzlerBFKN11}. Topological motif discover is also known to have many applications, in particular in protein networks \cite{ScottIKS06}.

Since its introduction, \textsc{Graph Motif} has received significant attention. It is known to be NP-hard even when the given graph is a tree of maximum degree 3 and the motif is a set~\cite{FellowsFHV11}. However in most practical cases the size of the targeted motif is relatively small, and so the parameterized version of the problem is arguably natural. The problem is fixed parameter tractable~\cite{FellowsFHV11} and the fastest known (randomized) parameterized algorithm~\cite{guillemot12finding} runs in $O^*(4^k)$ time \footnote{$O^*()$ hides factors polynomial in the input size.} and polynomial space, where $k$ is the size of the motif.

The algorithm of~\cite{guillemot12finding} is based on a reduction to the $k$-Multilinear Detection  problem (\textsc{$k$-MlD}) and a subsequent call of the fastest known algorithm for it \cite{Koutis08, WilliamsIPL09}. The \textsc{$k$-MlD} problem asks whether a polynomial presented as a circuit contains a multilinear term of degree $k$, when construed as a sum of monomials. Here, a circuit is a directed acyclic graph with addition and multiplication gates and terminals corresponding to variables.
 Notably, \textsc{$k$-MlD} gives also the fastest known algorithms for finding topological motifs such as $k$-trees \cite{KoutisW09}.

Motivated by \textsc{Graph Motif} we define the k - Constrained Multilinear Detection problem \textsc{$k$-CMlD}.

\textsc{$k$-CMlD}. On input of: \\
(i) an arithmetic circuit $C$ representing a polynomial $P(X)$ \\ (ii) a mapping $\chi:X\mapsto {\cal C}$ of variables $X$ to colors $\cal C$ \\ (iii) a mapping $\mu:{\cal C}\mapsto {\mathbb N}$ of colors to natural numbers (the multiplicities),

decide whether $P(X)$ construed as a sum of monomials contains an \textbf{allowed} multilinear monomial of degree $k$. An allowed multilinear monomial is a term that contains at most $\mu(c)$ appearances of variables colored with color~$c$.

The technique of~\cite{guillemot12finding} extends to the \textsc{$k$-CMlD} problem, giving an $O^*(4^k)$ time and polynomial space algorithm.  In this paper we show the following.

\begin{theorem} \label{th:main}
The \textsc{$k$-CMlD} problem can be solved by a randomized algorithm in $O^*(2.54^k)$ time and polynomial space.
\end{theorem}

We stress that this upper bound reflects a worst-case scenario which occurs only when the multiplicities $\mu(c)$ are all equal to 3. The running time can be as low as $O^*(2^k)$ when only $O(\log n)$ color multiplicities are different than~$1$.

We present the algorithm and prove its properties in Section \ref{sec:proof}. The algorithm is a simple modification of the \textsc{$k$-MlD} algorithm in \cite{WilliamsIPL09}. As a corollary, {\em any} problem that is reducible to \textsc{$k$-MlD} can now be solved with additional color constraints in  $O^*(2.54^k)$ time. This allows the combination of topological and functional constraints in motifs, for example by discovering colored $k$-paths. In Section~\ref{sec:reduction} we comment on applications of \textsc{$k$-CMlD} to \textsc{Graph Motif} and variants of it, such as the \textsc{Multiset Motif}, \textsc{Min-Add}, \textsc{Min-CC} and \textsc{Min-Substitute}. We obtain a faster parameterized algorithm for each of these problems.  
\section{The algorithm and its proof} \label{sec:proof}

\begin{proof}   We start by reviewing at a high level the purely algebraic algorithm for \textsc{$k$-MLD} in \cite{Koutis08, WilliamsIPL09}. The algorithm is based on properties of the group algebra $\Z2[\Z2^k]$ \cite{Koutis08}. The group multiplication over the group of $k$-dimensional 0-1 vectors is entry-wise addition modulo $2$. Consider any $k$-dimensional zero vector $v$, and let $v_0$ denote the identity of the group, that is the $0$ vector. In the group algebra $\Z2[\Z2^k]$, we have
$$
     (v_0+v)^2 = v_0^2 + 2v_0v +v^2 = 2v_0 +2v = 0~mod~2.
$$
To understand this identity, note the use of commutativity and of the facts $v_0^2=v^2=v_0$, $vv_0=v$.
This is the \textbf{annihilation property} which --by commutativity-- forces all non-multilinear terms of $P(X)$ to go to $0~mod~2$ if $C$ is evaluated on vectors of the form $v_0+v$. The second important property of this assignment is the \textbf{survival property} which is the fact that products of the form
$$
    (v_0 +v_1)(v_0 +v_2)\cdots(v_0 +v_k)
$$
evaluate to non-zero if and only if the vectors $\{v_1,\ldots,v_k\}$ are linearly independent over $\Z2$. In~\cite{Koutis08} it is also shown that $k$ random 0-1 vectors are linearly independent with probability at least $1/4$. So by assigning to each variable a random value of the form $v_0+v$ makes each multilinear monomial of $P(X)$ evaluate to non-zero with probability at least $1/4$.

The second part of the algorithm \cite{WilliamsIPL09} treats the case where a multilinear monomial has an even number of copies in $P(X)$. The algorithm first constructs an `extended' version $\tilde{{C}}$ of the input circuit $C$ which represents the same polynomial $P(X)$. It then labels each edge incoming to an addition gate in $\tilde{{C}}$ with a different multiplier from a set of additional variables $A$. By the construction of $\tilde{{C}}$ and the positions of the multipliers, $P(X,A)$ has the \textbf{isolation property}: the coefficient of each multilinear monomial in the polynomial $P(X,A)$ is $1$. The algorithm then assigns random values from the field ${\mathbb F}=GF(2^{3+\log_2 k})$ to the variables in $A$, values of the form $(v_0+v_i)$ to the variables in $X$ as described above, and evaluates $\tilde{{C}}$ over ${\mathbb F}[\Z2^k]$. Essentially because each copy of a given multilinear monomial now gets a random multiplier from $\mathbb F$, there is a good probability that their sum is not $0$~mod~2; formally this follows from the Schwartz-Zippel Lemma.  The output is `yes' if and only if the output of $\tilde{{C}}$ is non-zero, an event which happens with constant probability.

We now modify the algorithm in order to solve \textsc{$k$-CMlD}. The only change is in the assignment $X\mapsto \Z2[\Z2^k]$ which will now force not only non-multilinear but also `undesired' multilinear monomials to evaluate to $0~mod~2$.

\noindent (a) For each $c\in C$ pick $\mu(c)$ random `basis' vectors from $\Z2^k$. Let $S_c$ denote the subspace over $\Z2$ formed by these basis vectors. Repeat if the dimension $dim(S_c)$ of the subspace $S_c$ is less than $\mu(c)$. 

\noindent (b) For each $x\in X$ pick a random vector $v\neq v_0$
from $S_c$, where $c=\chi(c)$. Assign to $x$ the element $v_0+v$.

These two steps are easily implementable in $O(k n)$ time and space.  The claim for (a) is a corollary of the results in \cite{Koutis08}. Picking a random vector in (b) can be done by picking a random subset of vectors of the $S_c$ basis and adding them up.

It can be seen that with the proposed assignment $X\mapsto \Z2[\Z2^k]$  non-multilinear monomials evaluate to $0$ by the annihilation property. For the undesired multilinear monomials we will use the survival property but in its negative form: the product
$$
    (v_0 +v_1)(v_0 +v_2)\cdots(v_0 +v_t)
$$
is equal to $0~mod~2$ if the vectors $\{v_1,\ldots,v_t\}$ are linearly dependent over $\Z2$. A multilinear term which contains more than $\mu(c)$ variables colored with $c$ will evaluate to $0$ simply because the corresponding vectors are linearly dependent, as their number exceeds the dimension $\mu(c)$ of the subspace $S_c$ they belong to. In summary, we get a \textbf{constrained annihilation property}.

It remains to show that the assignment has a \textbf{constrained survival property} too. Consider one of the `allowed' multilinear terms. In general the term will consist of same-colored groups of variables $g_1,\ldots,g_m$, where group $g_i$ is colored with $c_i$ and contains $t_i\leq \mu(c_i)$ (distinct) variables. A necessary condition for the multilinear term to evaluate to non-zero is a  \textbf{joint survival property}: for all $i$, the product of the variables in group $g_i$ must evaluate to non-zero.

Using the analysis in \cite{Koutis08}\footnote{The random assignment in \cite{Koutis08} doesn't reject the $v_0$ vector. This makes the calculation in \cite{Koutis08} slightly different; the denominator is $2^t$ instead of $2^{t}-1$. Here it makes sense rejecting $v_0$ to improve the survival probabilty.}, the subspace survival probability for a group of size $t=\mu(c)\geq 2$ is precisely equal to
$$
    p_{t} = \prod_{i=1}^{t-1} \frac{2^t - 2^{t-i}}{2^{t}-1}.
$$
We also trivially have $p_1=1$. The product $\prod_{i=1}^m p_{t_i}$ is equal to the joint survival probability by independence of the events. We claim that, within a constant, the grouping of minimum joint survival probability is the
one that consists of $k/3$ groups of size $t=3$ each. To see why, note that for all $t$ we have $p_t>1/4$ and $p_2^4<1/4$. In addition we have $p_7>p_5p_2$, $p_6>p_3^2$, $p_5>p_3p_2$. This implies the minimum probability grouping cannot contain groups of size larger than 4, because they can be replaced with smaller groups to get a grouping with a lower probability. Similarly, $p_4^3>p_3^4$ implies that a minimum probability term cannot contain more than $2$ groups of size $4$. On the other hand $p_3^2<p_2^3$ and $p_2<p_1^2$  imply that the minimum probability grouping can't contain more than $2$ groups of size $2$ or $1$. Hence a term with minimum joint survival probability must consist of groups of size $3$ and at most $2$ groups of size $4,2$ and $1$.

It follows that the joint survival probability of any allowed multilinear term is $\Theta(p_3^{k/3})$.  Conditioned on the event of joint survival, the probability that the multilinear term evaluates to non-zero is at least $1/4$ by the usual survival property. This is because the subspaces are constructed independently at random; indeed the proof in \cite{Koutis08} looks at the extreme case when all groups are singletons.

Overall, if $P(X)$ contains an allowed term, evaluating $\tilde{C}$ on the proposed assignment returns a non-zero value with probability at least $p_3^{k/3}$. To boost this probability to a constant value we need to repeat $O((1/p_3)^{k/3})$ times. As in \cite{Koutis08,WilliamsIPL09}, each evaluation can be done in polynomial space and $O^*(2^k)$ time. So the overall running time is at most $O^*((8/p_3)^{k/3})$. After a simple calculation using the precise value for $p_3$, the running time turns out being around $O^*(2.54^{k})$. This finishes the proof.
\end{proof}

\textbf{Even Faster:} As it is clear from the proof, the number and sizes of the groups in an allowed multilinear term affects its probability of survival. The algorithm is in general faster when fewer color groups occur in the allowed term.   For example when the multilinear term contains  $k/10$ color groups the algorithm runs in $O^*(2.26^k)$ time.  The algorithm can be as fast as $O^*(2^k)$ when there are $O(\log n)$ non-singleton groups (the probability of subspace survival is inversely polynomial).

\section{\textsc{Graph Motif} and Variants} \label{sec:reduction}

We now present the consequences of our main result for the \textsc{Graph Motif} problem and some variants introduced more recently to account for noise in biological data \cite{DondiFV11b,DondiFV11}. In most cases the reductions to \textsc{$k$-CMlD} are identical to the reductions of the \textsc{Colorful Motif} problem to \textsc{$k$-MlD} given in \cite{guillemot12finding}; this is the case when every color appears exactly once in the motif. For this reason we omit most of the details. We note here that for each problem we state the result for its decision version. In every case the search problem can be solved via an easy reduction to the decision problem.

\begin{theorem} \label{th:motif}
 There is a randomized algorithm for \textsc{Graph Motif} that runs in time $O^*(2.54^k)$ and polynomial space.
\end{theorem}
\begin{proof}
We assign to each vertex $v$ of the input graph $G$ a variable $x_v \in X$. For a (connected) subtree $T=(V_T,E_T)$ of $G$ we let $Q_{T}= \prod_{v \in V_T} x_v$. It is straightforward (as shown  in \cite{guillemot12finding} in the algorithm for the \textsc{Colorful Motif} problem) to construct a circuit $C$ of size $O(k|G|)$ with the following properties: (i)~ Each  subtree of size $k$ contributes at least one copy of $Q(T)$ in the polynomial $P(X)$ represented by $C$.  (ii) Each multilinear monomial of degree $k$ in $P(X)$ is equal to $Q(T)$ for some size-$k$ subtree $T$ of $G$. We further assign the colors of the nodes to the corresponding variables and set the multiplicities $\mu(c)$ according to the motif. This gives an instance of  \textsc{$k$-CMlD}.
\end{proof}

In the \textsc{Graph Motif} problem we have that $\sum_{c\in C} \mu(c) = k$. The reduction to \textsc{$k$-CMlD} also solves a relaxed version of \textsc{Graph Motif} where the equality doesn't have to hold; this is known as \textsc{Max Graph Motif} \cite{DondiFV11b} or \textsc{Multiset Graph Motif} \cite{guillemot12finding}. As shown in~\cite{guillemot12finding} the reduction to \textsc{$k$-MlD} can be used to solve a relaxation dubbed \textsc{Graph Motif with Gaps} in \cite{guillemot12finding}; an alternative but equivalent formulation is called \textsc{Min Add} in~\cite{DondiFV11}. Informally this problem asks for a connected subgraph of size $r>k$ that contains a subset of the motif, for some specified $r$.  The reduction of \cite{guillemot12finding} extends immediately to  \textsc{$k$-CMlD}, improving the running time from $O^*(4^k)$ to $O^*(2.54^k)$.  The same improvement applies to the \textsc{Min-CC} problem, which asks for a minimum number of connected subgraphs that cover the motif, again via the reduction given in~\cite{guillemot12finding}.

We finally consider the \textsc{Min-Substitute} problem~\cite{DondiFV11}. Informally this problem asks for a connected subgraph with color multiplicities that are allowed to deviate from the motif, but as little as possible. Viewed alternatively, the subgraph meets the specifications in the motif but without accounting for $p$ of its nodes, where $p$ is as small as possible.  The fastest known algorithm runs in time  $O^*((3e)^k)$~\cite{DondiFV11}. We give an improved algorithm, which is not a direct reduction to \textsc{$k$-CMlD} but uses elements from it.

\begin{theorem}
 There is a randomized algorithm for \textsc{Min-Substitute} that runs in $O^*(5.08^k)$ time and polynomial space.
\end{theorem}
\begin{proof}
Let $P(X,A)$ be the `extended' polynomial from the proof of Theorem \ref{th:motif}; recall that it encodes the instance of the \textsc{Max Motif} for motifs of size $k$. We will introduce an extra set of variables $Y$ that are in 1-1 correspondence with the variables of $X$. We introduce them in the circuit by replacing each terminal $x_i$ with the gate $a_{1,i} x_i+ a_{2,i} y_i$, where $a_{1,i},a_{2,i}$ are fresh variables from the special set $A$. Note then that each multilinear term $x_{i_1}\cdots x_{i_k}$ generates $2^k$ multilinear terms (now in the variables $(X,Y)$). However the use of multipliers $a_{1,i},a_{2,i}$ ensures the isolation property, exactly as described in the proof of Theorem \ref{th:main}. We will assign to the variables in $A$ random values from ${\mathbb F}=GF(2^{\log k +5})$. The argument of \cite{WilliamsIPL09} applies identically.

Having fixed the assignment to $A$, we evaluate $P(A,X,Y)$ on a slightly more involved assignment over $\Z2[\Z2^{2k}]$. Let $(v_0+v_{x_i})$ be the assignment to the variable $x_i$ in the proof of Theorem \ref{th:main}; recall that this is a $k$-dimensional vector. Each variable $x_i$ gets assigned an element of the form
$$
     (v_0+\tilde{v}_x)(v_0+w_{x_i})
$$
and the corresponding variable $y_i$ gets assigned $z(v_0+w_{x_i})$, where:

\noindent (i) $\tilde{v}_{x_i}$ is the vector $v_{x_i}$ padded with zeros in the lower $k$ coordinates.
\\ (ii)  $w_{x_i}$ is a vector with zeros in the upper $k$ coordinates and a random 0-1 vector in the lower $k$ coordinates.
\\\  (iii) $z$ is a free indeterminate that will remain unevaluated and will be handled symbolically.

We will evaluate the extended circuit over ${\mathbb F}[z][\Z2^{2k}]$. Here the coefficients of the vectors in $\Z2^k$ are univariate polynomials from ${\mathbb F}[z]$.
Note that every non-multilinear term in $P(X)$ gives $P(A,X,Y)$ terms that are multiples of $x_i^2$ or $y_i^2$ or $x_i y_i$. When evaluated, these three monomials are multiples of $(v_0+w_{x_i})^2=0$. So, roughly speaking, the lower $k$ coordinates always enforce the annihilation property. Also note that the factors that use vectors non-zero in the upper $k$ coordinates enforce the constrained annihilation property, by construction. Assume however that a multilinear term $x_{i_1}\cdots x_{i_k}$ doesn't meet the multiplicities of the motif, unless it can avoid accounting for the multiplicities contributed by  $p$ of its variables/vertices (wlog the first $p$). Then the `upper coordinates' for the multilinear term $y_{i_1}\cdots y_{i_p} x_{i_{p+1}}\cdots x_{i_k}$ evaluate to a non-zero value multiplied by $z^p$, with probability at least $O((p_3)^{k/3})$, as in the proof of Theorem \ref{th:main}. This is because the $p$ variables are essentially dropped from the $k$ upper coordinates that enforce constrained annihilation. Independently from that, the lower coordinates evaluate to non-zero with probability at least $1/4$, by the survival property. Hence the coefficient of $z^p$ is non-zero (with some probability) if and only if there is a connected subgraph that exceeds by $p$ the motif multiplicities. So, finding the smallest such $p$ solves the problem. The evaluation of the circuit can be done in $O^*(4^k)$ time and polynomial space, and we need $O((1/p_3)^{k/3})$ evaluations. This finishes the proof.
\end{proof}


\begin{thebibliography}{10}

\bibitem{BetzlerBFKN11}
Nadja Betzler, Ren{\'e} van Bevern, Michael~R. Fellows, Christian Komusiewicz,
  and Rolf Niedermeier.
\newblock Parameterized algorithmics for finding connected motifs in biological
  networks.
\newblock {\em IEEE/ACM Trans. Comput. Biology Bioinform.}, 8(5):1296--1308,
  2011.

\bibitem{DondiFV11b}
Riccardo Dondi, Guillaume Fertin, and St{\'e}phane Vialette.
\newblock Complexity issues in vertex-colored graph pattern matching.
\newblock {\em J. Discrete Algorithms}, 9(1):82--99, 2011.

\bibitem{DondiFV11}
Riccardo Dondi, Guillaume Fertin, and St{\'e}phane Vialette.
\newblock Finding approximate and constrained motifs in graphs.
\newblock In Raffaele Giancarlo and Giovanni Manzini, editors, {\em CPM},
  volume 6661 of {\em Lecture Notes in Computer Science}, pages 388--401.
  Springer, 2011.

\bibitem{FellowsFHV11}
Michael~R. Fellows, Guillaume Fertin, Danny Hermelin, and St{\'e}phane
  Vialette.
\newblock Upper and lower bounds for finding connected motifs in vertex-colored
  graphs.
\newblock {\em J. Comput. Syst. Sci.}, 77(4):799--811, 2011.

\bibitem{guillemot12finding}
Sylvain Guillemot and Florian Sikora.
\newblock Finding and counting vertex-colored subtrees.
\newblock {\em Algorithmica}, pages 1--17.
\newblock 10.1007/s00453-011-9600-8.

\bibitem{Koutis08}
Ioannis Koutis.
\newblock Faster algebraic algorithms for path and packing problems.
\newblock In {\em Proceedings of the 35th International Colloquium on Automata,
  Languages and Programming, Part I}, pages 575--586, Berlin, Heidelberg, 2008.
  Springer-Verlag.

\bibitem{KoutisW09}
Ioannis Koutis and Ryan Williams.
\newblock Limits and applications of group algebras for parameterized problems.
\newblock In {\em Proceedings of the 36th International Colloquium on Automata,
  Languages and Programming: Part I}, ICALP '09, pages 653--664, Berlin,
  Heidelberg, 2009. Springer-Verlag.

\bibitem{LacroixFS06}
Vincent Lacroix, Cristina~G. Fernandes, and Marie-France Sagot.
\newblock Motif search in graphs: Application to metabolic networks.
\newblock {\em IEEE/ACM Trans. Comput. Biology Bioinform.}, 3(4):360--368,
  2006.

\bibitem{ScottIKS06}
Jacob Scott, Trey Ideker, Richard~M. Karp, and Roded Sharan.
\newblock Efficient algorithms for detecting signaling pathways in protein
  interaction networks.
\newblock {\em Journal of Computational Biology}, 13(2):133--144, 2006.

\bibitem{WilliamsIPL09}
Ryan Williams.
\newblock Finding paths of length $k$ in {$O^*(2^k)$} time.
\newblock {\em Inf. Process. Lett.}, 109:315--318, February 2009.

\end{thebibliography}
\end{document}